\documentclass[12pt, draftclsnofoot, onecolumn]{IEEEtran} 
\usepackage{lipsum}
\usepackage{amsmath}
\usepackage{graphicx}
\usepackage[export]{adjustbox}
\usepackage{tikz}
\usepackage{pgfplots}
\usepackage{amssymb}
\usepackage[caption=false]{subfig}
\usepackage[utf8]{inputenc}
\usetikzlibrary{decorations.pathreplacing}
\graphicspath{{images/}}
\usetikzlibrary{patterns}
\usepgfplotslibrary{fillbetween}
\pgfplotsset{compat=1.14,width=10cm}
\usepackage{cite}
\usepackage{optidef}
\usepackage{amsthm}
\usepackage[ruled,linesnumbered]{algorithm2e}
\usepackage{setspace}
\newtheorem{theorem}{Theorem}
\usepackage{booktabs}
\usepackage{multirow}

\title{Energy Optimization in Ultra-Dense Radio Access Networks via Traffic-Aware Cell Switching }

\author{Metin Ozturk\IEEEauthorrefmark{1}, Attai Ibrahim Abubakar\IEEEauthorrefmark{1}, João Pedro Battistella Nadas\IEEEauthorrefmark{1}, Rao Naveed Bin Rais\IEEEauthorrefmark{2}, Sajjad Hussain\IEEEauthorrefmark{1} and Muhammad Ali Imran\IEEEauthorrefmark{1}
\thanks{\IEEEauthorrefmark{1}Communication, sensing and imaging (CSI) research group, James Watt School of Engineering, University of Glasgow, United Kingdom. 
Emails: m.ozturk.1@research.gla.ac.uk, a.abubakar.1@research.gla.ac.uk, j.battitsella-nadas.1@research.gla.ac.uk, sajjad.hussain@glasgow.ac.uk, and muhammad.imran@glasgow.ac.uk
\IEEEauthorrefmark{2}Electrical and Computer Engineering, Ajman University, UAE. Email: r.rais@ajman.ac.ae}
}

\begin{document}

\maketitle

\begin{abstract}
Ultra-dense deployments in 5G, the next generation of cellular networks, are an alternative to provide ultra-high throughput by bringing the users closer to the base stations.
On the other hand, 5G deployments must not incur a large increase in energy consumption in order to keep them cost-effective and most importantly to reduce the carbon footprint of cellular networks.
We propose a reinforcement learning cell switching algorithm, to minimize the energy consumption in ultra-dense deployments without compromising the quality of service~(QoS) experienced by the users.
In this regard, the proposed algorithm can intelligently learn which small cells~(SCs) to turn off at any given time based on the traffic load of the SCs and the macro cell.
To validate the idea, we used the open call detail record~(CDR) data set from the city of Milan, Italy, and tested our algorithm against typical operational benchmark solutions.
With the obtained results, we demonstrate exactly when and how the proposed algorithm can provide energy savings, and moreover how this happens without reducing QoS of users.
Most importantly, we show that our solution has a very similar performance to the exhaustive search, with the advantage of being scalable and less complex.
\end{abstract}

\section{Introduction}

The energy consumption of mobile cellular networks (MCNs) is escalating greatly as a result of the exponential increase in data demands thereby forcing mobile network operators~(MNOs) to continually expand their networks by deploying more base stations~(BSs) to accommodate this data crunch.
The increase in data demands is due to the ever increasing number of hand-held and machine-type devices as well as the development of data-hungry mobile applications such as video-conferencing, streaming, on-line gaming, etc~\cite{FengONOFF}. 

In addition, the provision of local support for Internet of Things (IoT) in the fifth generation of cellular networks~(5G) in order to meet the growing demands for applications including smart health, smart homes, and smart cities, would result in the massive deployment of IoT devices~\cite{Li2018}, which leads to further increases in the energy consumption of cellular networks.

Network densification, comprising the ultra-dense deployment of small cells (SCs) alongside macro cells (MCs) has been identified as one of the enablers for achieving the one thousand fold capacity increase in 5G networks~\cite{Andrews2014, GeUDN}. 
SCs are low power nodes which are deployed to boost the capacity at hotspot zones, the busiest areas of a deployment.
However, the massive deployment of SCs would also result in a dramatic increase in energy consumption.

Two adverse effects are rooted on the increasing energy consumption in MCNs: 1) An increase in green house gas emission to the environment and 2) an increase in the network operating cost due to higher energy bills charged to MNOs~\cite{Buzzi2016}.
Therefore, energy optimization needs to be given more consideration in 5G networks so that the increase in network capacity due to network densification does not incur large increase in energy consumption, thereby making it cost-effective and most importantly environmentally friendly~\cite{ALAMU2020}.

BSs are the major energy consumers, accounting for about 60\%--80\% of the total energy consumption in cellular networks~\cite{FengONOFF}.
As such, reducing the energy consumption of BSs will amount to significant reduction in the total energy consumption of MCNs.
From energy optimization perspective, the ideal operation of the BS is to be in such a way that the energy consumption scales with traffic load, meaning that at no load, the power consumption should be negligible and gradually increase with the increasing load. 
However, in practice that is not the case, as BSs still consume a significant amount of energy (about 50\%--60\% of its maximum energy consumption) when not serving any user~\cite{Buzzi2016}. 
Since the traffic load of cellular networks exhibit temporal and spatial variations, the traditional technique of keeping the BS always ON even when it is not serving any user results in energy wastage. 
Therefore, load adaptive network operation, where BSs are turned off or operated in low power modes during periods of low or no traffic in order to save energy, has been the focus of many studies~\cite{Luo2018, Sun2018, Wu2016, Buzzi2016, YuEE, FengONOFF}.

Nevertheless, it is not always feasible to completely switch off SCs in the conventional heterogeneous network~(HetNet) architecture because it often creates coverage holes, which in turn degrade the Quality of service (QoS) of users initially covered by the inactive SCs. 
In addition, sleeping BSs do not transmit pilot signals needed by the user equipment (UE) for cell discovery, channel estimation, and subsequent connection, hence in the conventional HetNet architecture certain components of the SCs need to be left always ON even in sleep mode, resulting in sub-optimal energy savings~\cite{mohamed2016control, onireti2015energy}.

Handling the aforementioned challenges of BS switching in conventional HetNets requires a paradigm shift towards a Control/Data Separated Architecture (CDSA)~\cite{mohamed2016control}. 
In CDSA, the MCs---also known as control BS (CBS)---maintain constant coverage and provide signalling functionalities and low data rate services, while the SCs---also known as data BS (DBS)---provide high data rate services and are connected to the MCs through the backhaul. 
This architecture provides support and flexibility for dynamic cell switching operations as the MC always ensures constant coverage for both idle and active users, and is responsible for switching BSs off/on as well as associating users to the SCs, thereby making the complete switch off of DBSs possible~\cite{mohamed2016control}.

Cell switching with traffic offloading has been identified as one of the techniques for reducing the energy consumption of MCNs.
Several methods have already been proposed  for scheduling cell switching in the literature using analytical modelling and heuristic algorithms~\cite{Luo2018, Sun2018, Wu2016, Buzzi2016}. 
However, it is very difficult to develop accurate analytical models for network optimization when network dimensions become very large due to network complexity and high computational overhead~\cite{tabassum2018mobility}. 
On the other hand, heuristic algorithms are difficult to generalise and adapt to dynamic environment, such as ultra-dense network deployment scenarios that have been envisioned in 5G~\cite{bui2017survey, ALAMU2020, HussainMLsurvey}. 
They also mostly employ exhaustive search techniques, which makes them computationally demanding and could lead to degradation in the QoS of users.
Recent works in~\cite{Kong2013, el2019distributed, salem2018reinforcement, abubakar2019, el2019location, park2020reinforcement} applied conventional reinforcement learning~(RL) techniques ($Q$-learning) for cell switching because of its ability to adapt to a dynamic network environment through learning. 

However, conventional RL algorithms are very challenging to implement when the network dimensions become huge because it often results in very large state-action space, wh5Gich is computationally demanding to learn. In addition, a considerable amount of memory is required to store the action-value table ($Q$-table).
This makes such algorithms impracticable to implement in real networks. 
In an attempt to solve the curse of dimensionality problem facing conventional RL algorithms, they were combined with value function approximation in order to estimate the optimal policy. 
In this regard, RL with linear function approximators and deep RL approaches were proposed in~\cite{Chen2015, liu2018deepnap, ye2018drag}. Deep RL algorithms have the ability to accommodate large state-action space resulting from large scale networks deployment scenarios, however, training such deep neural network models can be computationally demanding, energy consuming. Hence, they should only be considered when no simpler solution approaches are available or when the complexity of the network requires the application of a non-linear function approximator to estimate the optimal policy. 

We propose a RL algorithm with linear function approximation known as State Action Reward State Action (SARSA) with value function approximation~(VFA) for cell switching and traffic offloading in ultra-dense radio access network~(RAN). In the proposed algorithm, all the states need not be visited as in~\cite{Chen2015} before the optimal policy is learnt. In addition, the algorithm exhibits quick convergence and it is simpler to implement compared to deep RL approaches. 
The learning algorithm is implemented at each MC and it has the ability to learn the optimal switching pattern even when a large number of SCs are deployed under the coverage area of the MCs.

\subsection{Related work}
The authors in~\cite{Luo2018} developed a load based dynamic SC switching scheme for ultra-dense networks in order to minimize the signalling overhead resulting from user traffic offloading during the cell switching process as well as to optimize the energy savings of the network using two heuristic algorithms. 
In~\cite{Sun2018}, the problem of user association and cell sleeping in multi-tier ultra-dense SC networks was formulated as a complex integer programming. Then, two low complexity heuristic algorithms were employed to determine the optimal user association and the cell switching pattern.
A greedy heuristic algorithm was proposed in~\cite{Wu2016} to determine the switching off/on pattern of SCs in a green ultra-dense network in order to optimize the network energy efficiency by considering traffic load of the SCs and service requirements of users.
Heuristic algorithms often employ exhaustive search, which is often slow and computationally demanding, to find the optimal solution, hence they are only suitable for small network deployment. 
On the other hand, the ultra-dense 5G network scenario will involve massive deployment of SCs which would make it practically impossible to adopt such heuristic algorithms, as it would result in huge computational overhead and degradation in QoS.   

An alternative solution for finding optimal switching pattern for ultra-dense deployment scenarios is to employ RL techniques. 
The authors in~\cite{Kong2013, el2019distributed, salem2018reinforcement, abubakar2019, el2019location,  Chen2015} proposed $Q$-learning based cell switching techniques for energy optimization. 
In~\cite{Kong2013}, the BS was assumed to comprise of modular resources, and a $Q$-learning algorithm was developed to dynamically activate and deactivate certain number of modules in the BS based on the traffic demand at each time instance. 
A distributed $Q$-learning algorithm was proposed in~\cite{el2019distributed} to control the sleep depth of the BS in order to minimize energy consumption and network latency.  
The authors in~\cite{salem2018reinforcement} proposed a $Q$-learning algorithm to determine the duration of time that the BS can spend at a particular sleep level in order to optimize energy consumption of the network while considering BS activation latency and the service requirements of users as constraints. 
A $Q$-learning framework was developed in~\cite{abubakar2019} in order to determine the optimal switching and traffic offloading strategy in a two-tier heterogeneous network with separation architecture. 
The authors in~\cite{el2019location} proposed a location-aware multi-level sleep mode strategy using $Q$-learning to determine the sleep mode level of the BS based on user locations in the network and their velocity towards the neighbouring BS.
Nonetheless, the works in~\cite{Kong2013, el2019distributed, salem2018reinforcement, abubakar2019, el2019location} only considered small to medium network deployment scenarios, where the state-action space is suitable for the implementation of conventional RL algorithms.

The authors in~\cite{Chen2015} proposed a centralized and decentralised $Q$-learning algorithm with compact state representation ($Q$C-learning) for traffic offloading and cell switching for HetNet to minimize energy consumption. 
The centralized $Q$C-learning was implemented at the BS controller. 
The $Q$C-learning is a compact representation of the state-action pair using linear VFA when it becomes practically impossible to explicitly store each state-action pair in a look-up table for situations, where the number of SCs under the BS controller become very large.
Moreover, even with the compact state representation of the $Q$-value lookup table, as the number of SCs in the network becomes very large, the action set also grows dramatically, thereby making it difficult to implement a centralized cell switching and traffic offloading scheme. 
Therefore, a decentralized multi-agent-$Q$C-learning was developed, where the MCs learn in a cooperative manner and take joint traffic offloading and cell switching actions by exploiting the previous cell switching strategies used by other MCs. 
Furthermore, in developing decentralized multi-agent-$Q$C-algorithm, the authors~\cite{Chen2015} assumed that all the MCs under the controller have similar or stationary network states, which is a requirement for implementing joint cell switching and traffic offloading strategy.
This might not be the case in real networks, as networks' states may vary from one MC to another due to temporal and spatial variations in user traffic demands~\cite{Busoniu2008}, and as such it might not always be feasible to perform joint cell switching and traffic offloading actions. 
In addition, the problem of increased state-action pair also arises in multi-agent $Q$-learning, since each agent also includes its own state-action pair to the joint state action space~\cite{Busoniu2008}.
This increases the computational complexity as well as the memory required for storing the joint state-action space at the controller.
\subsection{Contribution}
In this article, we propose an intelligent cell switching and traffic offloading framework using a RL technique known as SARSA with VFA~\cite{Sutton2018} in order to reduce the energy consumption of ultra-dense RAN\@.
In the CDSA RAN, a cell switching and traffic offloading mechanism is implemented in a locally centralized manner at each MC, which is responsible for scheduling the switching off/on pattern of all the SCs deployed under its coverage.
The proposed algorithm provides a compact form of representing the action-value function like the $Q$C-learning algorithm in~\cite{Chen2015}, but all the states need not be visited as in~\cite{Chen2015} before the optimal strategy is obtained.
The contributions of this work are as follows:
\begin{itemize}
    \item We propose a scalable traffic-aware RL algorithm based on VFA to find the optimal policy in terms of energy minimization for controlling SC ON/OFF status, without compromising the QoS.
    \item We prove mathematically that turning a SC off is not always profitable in terms of energy savings, and identify such situations where it is not profitable to switch off SCs. 
    \item We evaluate the proposed algorithm using a traffic model based on real world data, making the solution more reliable and realistic.
\end{itemize}

Due to its fast convergence and learning ability, RL is employed in this work to make real-time, accurate, and efficient switch off/on decision at each time slot.
Moreover, VFA is utilized in the developed RL algorithm, since the state space expands exponentially with increasing number of SCs in the network.
Even though this could be handled by conventional RL algorithms to some extent, it becomes infeasible to manage once the network size becomes very large.
In addition, due to the careful and proper design of the action set in the proposed algorithm, there is no need to include all the possible switching combinations, thus annihilating the need for cooperative learning.
Lastly, we tested the proposed algorithm in a realistic scenario, where all types of SCs given in~\cite{Auer2011} with their diverse characteristics are included.

The remaining parts of this paper are organized as follows: the system model is presented Section~\ref{sec:model}, while Section~\ref{sec:problem} presents the problem formulation and the derivation of an analytical model for SC switching.
The proposed intelligent solution is discussed in Section~\ref{sec:proposed}, and Section~\ref{sec:performance} evaluates the performance of the proposed  solution and the developed benchmark methods.
Lastly, Section~\ref{sec:conclusion} concludes the paper.

\section{System Model}\label{sec:model}
\subsection{Network Model}
As mentioned before, in this work we consider an ultra-dense RAN with a CDSA architecture~\cite{mohamed2016control}. 
The network model, as illustrated in Fig.~\ref{fig:model}, consists of a dense network, where SCs, acting as DBSs, are deployed under the coverage area of a MC, which acts as the CBS\@.
Moreover, SCs and MC operate on dedicated frequency channels and SCs are connected to the MC via optical fibre links.

The MC is responsible for constant coverage, control signalling, and data services, while SCs handle only data services and user specific requests.
Furthermore, the MC coordinates the traffic offloading and switching off/on of all SCs under its coverage by observing their traffic loads and deciding which set of SCs should be turned off during periods of low traffic intensity, taking the available capacity of the MC into account.

\begin{figure}
\centering
\includegraphics[trim={.1cm 0.1cm 0.1cm 0.1cm},clip,width=.5\columnwidth]{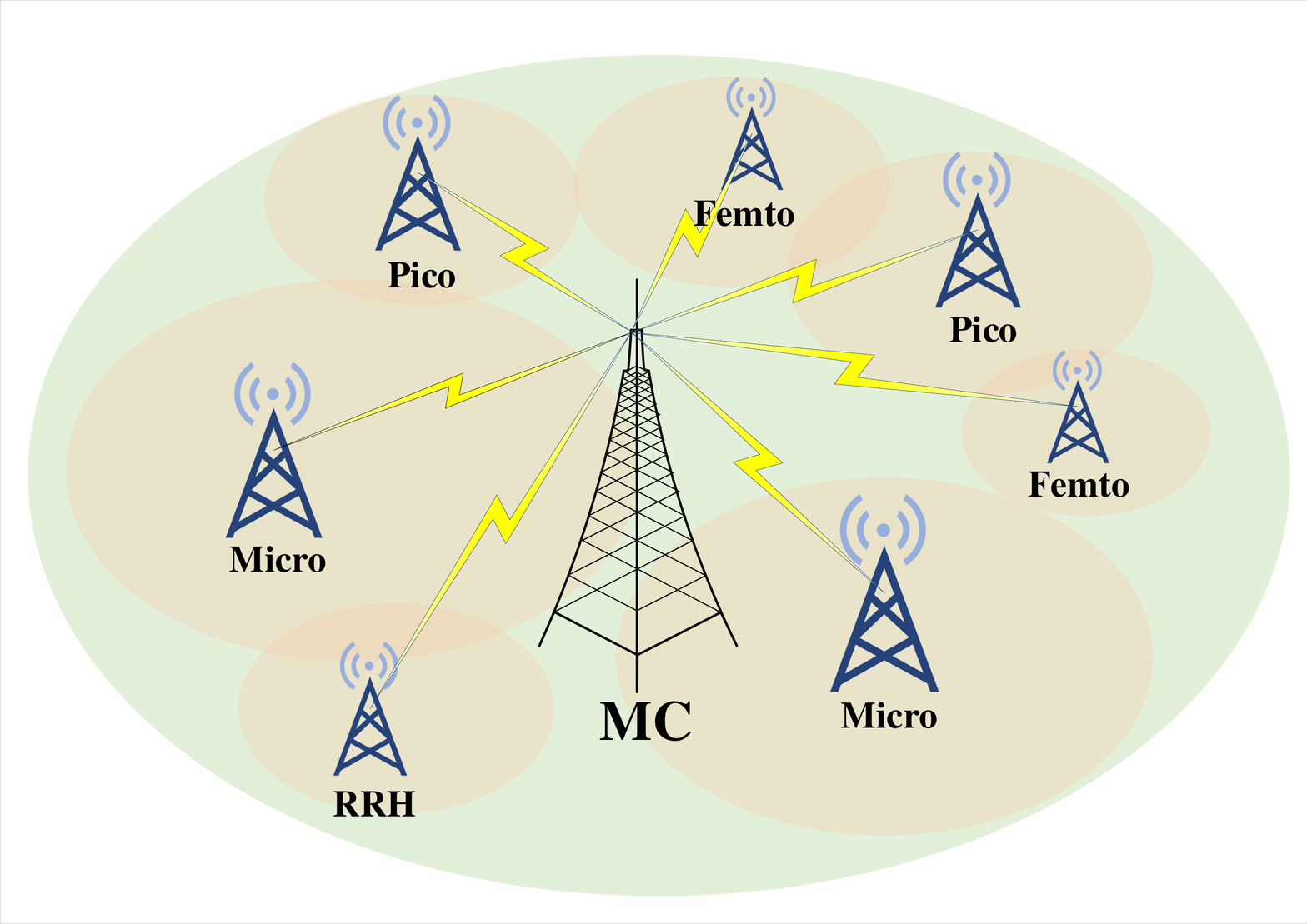}
\caption{A HetNet with CDSA comprising of a MC that works as the CBS and a dense deployment of various types of SCs~(micro, RRH, pico, femto) to provide DBS services.}\label{fig:model}
\end{figure}
\subsection{Network Power Consumption}
Following the Energy Aware Radio and neTwork tecHnologies~(EARTH) power consumption model~\cite{Auer2011}, $P_j$, the instantaneous consumption of a BS, $B_j$, is given by~\cite{Wu2017a}
\begin{equation}\label{eq:earthPower}
    P_j  = 
    \left\{
        \begin{array}{cc}
            P_{\text{o},j} + \eta_j \lambda_j P_{\text{T},j} & 0 < \lambda_j < 1\\
            P_{\text{s},j}                             &   \lambda_j=0,
         \end{array}
    \right.
\end{equation}
where $P_{\text{o},j}$ and $P_{\text{s},j}$ are the operational and sleep circuit power consumption, respectively, $\eta_j$ is the power amplifier (PA) efficiency, $\lambda_j$ is the load factor, and $P_{\text{T},j}$ is the transmit power.
Moreover, all SCs of the same type\footnote{There are five different types of BS considered in this work, macro, micro, femto, pico and remote radio head (RRH).} are considered to have identical hardware, such that their PA efficiencies, and circuit power consumptions are the same.
Also, power allocation is not considered, hence, each type of BS has a fixed transmit power that is constant among BSs of that type.

Lastly, $P$, the instantaneous power consumption for the considered CDSA network, is expressed as
\begin{equation}
\label{eq:totalPower}
P = \sum_{j=1}^{s+1} P_j,
\end{equation}
where $s$ is the number of SCs in the deployment.
\section{Problem formulation}\label{sec:problem}
Considering the architecture described above, the aim of this paper is to find the best policy, in terms of energy savings, which offers a required QoS to the users.
A policy $\pi$ is defined by which SCs should be ON at a given time $t$. 
In other words
\begin{equation}
    \pi = \{\delta_1,\delta_2,\dots,\delta_{s+1}\},
\end{equation}
where $\delta_j\in\{0,1\}$ indicates the state of $B_j$, 1 for ON and 0 for OFF.
$B_1$ represents the MC, and thus $\delta_1$ is always 1, as it is always ON.

Considering $j > 1$, when $\delta_j$ changes from 1 to 0 at time $t$, the MC allocates its users, such that
\begin{align}
    \lambda_{1,t} = \lambda_{1,t-1} + \phi_j\lambda_{j, t-1} \\
    \lambda_{j,t} = 0,
\end{align}
where $\lambda_{j, t}$ corresponds to the load of $B_j$ at time $t$, and $\phi_j$ is the relative capacity of $B_j$ with respect to $B_1$, such that
\begin{equation}
    \phi_j = \frac{C_j}{C_1}, \quad j>1,
\end{equation}
where $C_j$ indicates the total resources available in $B_j$.
Conversely, when $\delta_j$ switches from 0 to 1 at time $t$, the MC offloads some of its traffic to $B_j$, such that
\begin{align}
    \lambda_{j,t} = \frac{\tau_j}{C_j}\label{eq:lambdasc1} \\
    \lambda_{1,t} = \lambda_{1,t-1} - \phi_j\lambda_{j, t}\label{eq:lambdamc1},
\end{align}
where $\tau_j$ corresponds to the resources used by users served by $B_j$.
Note that $C_j \geq \tau_j$.

Therefore, we can formally write the problem as
\begin{mini}|s|
    {\pi}{P(\pi)} 
    {\label{eq:opt}}{}
    \addConstraint{\lambda_1}{\leq 1.} 
\end{mini}
Note that the only constraint in the problem is to ensure that the capacity of the MC is not exceeded, this takes care of the QoS requirement.
In other words, the MC only offloads users from a SC if it can maintain the QoS of any user associated to it.
Moreover, this ensures that once a solution is implemented, 
\begin{equation}
T_{\text{r},1} \leq T_{\text{p},1},    
\end{equation}
where $T_{\text{r},1}$ and $T_{\text{p},1}$ are the required and provided throughputs by the MC, respectively.

\begin{theorem}\label{th:prob}
    If we draw a random SC $B_j$ from $\mathbb{B}$, the set of all possible BSs, the probability of $\delta_j = 1$ integrating the optimal policy is 1 if $j>1$ and
    \begin{equation}\label{eq:theorem}
        \frac{P_{\text{o},j} - P_{\text{s},j}}{\phi_j\eta_1 P_{\text{T},1} - \eta_j P_{\text{T}, j}} < \lambda_j,
    \end{equation}
    when $\phi_j\eta_1 P_{\text{T},1} - \eta_j P_{\text{T}, j} > 0$.
\end{theorem}
\begin{proof}
    Using~\eqref{eq:totalPower}, the difference in power consumption $\Delta P$ when considering changing $\delta_j$ to 1 can be expressed as
    \begin{equation}\label{eq:deltap}
        \Delta P = P_{t} - P_{t-1} = \sum_{i=1}^{s+1} P_{i,t} - \sum_{i=1}^{s+1} P_{i, t-1},
    \end{equation}
    where $P_t$ is the total power consumption of the network at time $t$ and $P_{i,t}$ is the power consumption of $B_i$ at time $t$.

    Next,~\eqref{eq:deltap} can be expanded as
    \begin{align}\label{eq:deltap2}
        \begin{split}
            \Delta P = P_{1,t} &+ P_{j, t} + \sum_{i=2, i \neq j}^{s+1} P_{i,t} \\ & - \left(P_{1,t-1} + P_{j,t-1} + \sum_{i=2, i \neq j}^{s+1} P_{i,t-1}\right).
        \end{split}
    \end{align}

    Assuming $P_{i\notin \{1, j\},t} =P_{i\notin \{1, j\},t-1}$, meaning all other SCs are kept at their states,~\eqref{eq:deltap2} then becomes
        \begin{equation}\label{eq:deltap3}
            \Delta P = P_{1,t} + P_{j, t} - P_{1,t-1} - P_{j, t-1}.
        \end{equation}
    
        Next, using~\eqref{eq:earthPower} and replacing~\eqref{eq:lambdasc1} and~\eqref{eq:lambdamc1} into~\eqref{eq:deltap3} yields
        \begin{equation}
            \begin{split}
                \Delta P = & P_{\text{o},1} + \eta_1 ( \lambda_{1,t-1} - \phi_j\lambda_{j, t}) P_{\text{T,1}} \\ 
                         + & P_{\text{o},j} + \eta_j \lambda_{j,t} P_{\text{T},j} \\
                         - & P_{\text{o},1} - \eta_1 \lambda_{1,t-1} P_{\text{T,1}}\\
                         - & P_{\text{s},j},
            \end{split}
        \end{equation}
        which can be further simplified to
        \begin{equation}\label{eq:deltap4}
            \Delta P =  P_{\text{o},j} + \eta_j \lambda_{j,t} P_{\text{T},j} - \eta_1 \phi_j\lambda_{j, t} P_{\text{T,1}} - P_{\text{s},j}.
        \end{equation}
    
        From~\eqref{eq:deltap4}, it is easy to see that $\Delta P < 0$ when 
        \begin{equation}\label{eq:inequality}
            P_{\text{o},j} + \eta_j\lambda_{j,t} P_{\text{T}, j} < \phi_j\eta_1\lambda_{j,t} P_{\text{T},1} + P_{\text{s},j} ,
        \end{equation}
        Note that, in order to isolate $\lambda_{j,t}$ in~\eqref{eq:inequality}, we must divide both sides by $\phi_j\eta_1 P_{\text{T},1} - \eta_j P_{\text{T}, j}$, and thus~\eqref{eq:theorem} is only valid for $\phi_j\eta_1 P_{\text{T},1} - \eta_j P_{\text{T}, j} > 0$.
        Lastly, we remove the index $t$ from~\eqref{eq:inequality} to make it general and solve for $\lambda_j$, yielding~\eqref{eq:theorem}.
\end{proof}

Based on Theorem~\ref{th:prob}, we can see that the optimal policy will tend to have more SCs turned on when they are more loaded, as the inequality in~\eqref{eq:theorem} will be more easily met.
Similarly, we can also see that when the transmit power of SCs is smaller, it becomes more advantageous to use them, and the same can be said if the MC is not very efficient ($\eta_1$ is small).
Moreover, as one would expect, when the power consumption of a sleeping BS is higher, the less likely it will be that turning it off would result in energy savings.
Lastly, we can clearly see that when~\eqref{eq:theorem} occurs for any BS in the network, there will be a situation, where the most optimal policy has SCs turned on and consumes less energy than keeping only the MC operational.
\section{Proposed Solution}\label{sec:proposed}
\subsection{Value Function Approximation}
The goal of this paper is to find the best policy, that is, the set of SCs to switch off/on per time, in order to optimize the energy consumption of the network.
However, since the complexity of the problem is exponentially increasing, in terms of possible configurations, its solution is non-trivial.
With that in mind, we propose an intelligent solution, based on RL, in order to minimize the energy consumption of the network without compromising the QoS.

However, traditional RL techniques---such as vanilla $Q$-learning or SARSA---rely on a state-action matrix, which has to contain an entry for each possible state-action pair.
This is an issue in this case because of the aforementioned dimensionality of the problem, which results in a very large matrix, and this in turn makes the algorithm infeasible.
To circumvent this issue, we use RL alongside VFA, which is capable of estimating the state based on a set of features, and therefore does not have to maintain a huge structure with all the possibilities~\cite{Sutton2018}.

The way VFA can be used to enhance traditional RL is by estimating the action-value function with an approximator.
It works by finding a set of weights for a known function by solving an optimization problem based on given examples.
In other words, VFA solves
\begin{mini}|s|
    {\theta}{|Q - \hat{Q}(\theta)|} 
    {\label{eq:vfaopt}}{}
\end{mini}
where $Q$ is the value of the matrix obtained from given examples, and $\hat{Q}$ is the estimation, expressed by~\cite{Sutton2018}
\begin{equation}
    \hat{Q} = f(\theta),
\end{equation}
where $f(.)$ is the known function---also referred to as hypothesis---, and $\theta$ is the set of weights.
\subsubsection{Linear Function Approximators}
Even though any function can be used as the hypothesis, some render~\eqref{eq:vfaopt} either too hard or infeasible.
Linear functions, shallow neural networks, and deep neural networks are commonly used as $f(.)$, and the choice depends on the type of problem.
With regards to the optimization, a popular strategy is to use gradient descent~(GD) or stochastic GD~(SGD) to find $\theta$ based on the known examples~\cite{Sutton2018}.

In this work, we use a linear hypothesis, meaning that $\hat{Q}$ is approximated by a linear combination of input features~\cite{Sutton2018}, in other words,
\begin{equation}\label{eq:qtable}
    \hat{Q} = \vec{X}\vec{\theta^T},
\end{equation}
where $\vec{X}$ and $\vec{\theta}$ are vectors containing input features and weights, respectively, and $\vec{\theta^T}$ indicates $\vec{\theta}$ transposed.
Furthermore, we use SGD as the optimizer, such that $\vec{\theta}$ is updated according to 
\begin{equation}
    \nabla_\theta{(Q(\vec{X}) - \hat{Q}(\vec{X}, \vec{\theta}))}^2,
\end{equation}
where $\nabla_\theta$ represents the gradient with respect to $\theta$.
In other words, at each iteration $t$ we perform~\cite{Sutton2018}
\begin{equation}\label{eq:weight_update}
   \vec{\theta}_{t+1} = \vec{\theta}_{t} - \frac{1}{2}\nabla_\theta{\left[Q(\vec{X}) - \hat{Q}(\vec{X}, \vec{\theta_t})\right]}^2.
\end{equation}
Note that, $\vec{\theta}$ is initialized with an arbitrary value, such as zero.

After some iterations, if enough different examples have been observed, $\vec{\theta}$ will have converged and can be used to find the action-value function, which in turn can guide the policy.

\subsection{SARSA with VFA}
For the readers benefit, we have included a small summary of SARSA in this section.
The algorithm works by observing the cost of taking actions and updating its estimate of $\vec\theta$ at every iteration and then choosing the best action according to its estimate of the action-value matrix.
Algorithm~\ref{alg:sarsavfa}~ \cite{Sutton2018} contains a pseudo code implementation of SARSA with VFA\@, where $S_{t}$ and $S_{t+1}$ are the current and next states, respectively; and $a_{t}$ and $a_{t+1}$ are the current and next actions, respectively.
Note that lines 22 to 29 in Algorithm~\ref{alg:sarsavfa} are correspondent to the stopping criteria, where $j_\text{min}$ is the minimum number of iterations to take before the stopping criteria come into the effect.
$C_\text{min}$ and $C_\text{max}$ are the minimum and maximum cost functions observed up to that iteration, respectively.
Lastly, $\Omega$ is the threshold for the feature scaled cost, while $j_\text{rep}$ defines the number of iterations to be repeated~(the conditions on lines 23 and 24 are kept satisfied) before stopping.


\begin{singlespace}

\begin{algorithm}[]
	\For{Every episode} 
	{
		Initialize the current state, $S_{t}$; \\
		\For{All actions in $\vec{A}$}
		{
			Get features, $\vec{X}$; \\
			Estimate value of $Q$ through~\eqref{eq:qtable}; \\
		}
		Choose action, $a_{t}$, according to policy; \\
		Set $j_\text{it}=0$;\\
		\For{Each iteration}
		{
		    Update $j_\text{it} \leftarrow j_\text{it}+1$;\\
			Take the action $a_{t}$; \\
			Observe cost, $C$, via~\eqref{eq:vfa_cost}; \\
			Move to next state; \\
			\For{All actions in $\vec{A}$}
			{
				Get features, $\vec{X}$; \\
				Estimate value of $Q$ through~\eqref{eq:qtable}; \\
			}
			Choose next action, $a_{t+1}$, according to policy;\\
			Update the weights, $\vec{\theta}_{t}$, using~\eqref{eq:weight_update}; \\
			$S_t \leftarrow S_{t+1}$; \\
			$a_t \leftarrow a_{t+1}$; \\
			\If{$j_\text{it} > j_\text{min}$}{
			    \If{$\dfrac{C-C_\text{min}}{C_\text{max}-C_\text{min}}\leq \Omega$}{
			        \If{$C$ is the same for $j_\text{rep}$ iterations}{
			            Stop executing\;
			            Jump to the next episode\;
			        }
			    }
			}
		}
	}
	\caption{Proposed SARSA with VFA}
	\label{alg:sarsavfa}
\end{algorithm}
\end{singlespace}
\subsection{Traffic Aware Energy Optimization via Cell Switching}
Leveraging the framework described above, we propose a linear VFA solution to solve~\eqref{eq:opt}.
Since we are looking to accomplish a globally optimal solution, our proposed framework is centralized and computed at the MC\@.
As the total number of policies increases exponentially with $s$, it would not be scalable to consider any policy as an action.
Therefore, we propose a reduced action space as follows:
\subsubsection{Actions}
The actions for the proposed VFA consist of switching off/on different SCs in the network.
However, because there are so many possibilities, we propose an alternative representation which allows the algorithm to sample several different possibilities by taking different actions. 

This representation is done as follows.
First, the status of the SCs in the network are converted to a binary number, such that the SCs that are ON are treated as binary 1, while the SCs that are OFF are considered binary 0.
In this regard, the status of the network at time $t$ is
\begin{equation}
    \vec{\chi}(t)=\{\chi_i(t) | i\in\{1,2,...,s\}\},
\end{equation}
where  $\chi_i\in\{0,1\}$ is the state of the $i^\text{th}$ SC in the network.

Next, $\vec{\chi}$ is represented by a binary number $\chi_\text{b}$ with $s$ digits, such that the status of each SC represents one of its digits and thus
\begin{equation}
    \chi_\text{b}=\chi_1\chi_2...\chi_s.
\end{equation}

Within the proposed representation, the set of possible actions is defined as 
\begin{equation}
    \vec{A}=[0, \pm\xi^0, \pm\xi^1,...,\pm\xi^s],
\end{equation}
where $\xi$ is a decimal constant number, defining the inter-space between two consecutive actions.

In this case, taking an action at time $t$, means to perform 
\begin{equation}
    \chi_\text{d}(t+1)=\chi_\text{d}(t)+A_z,
\end{equation}
where $A_z$ is an entry\footnote{Note that, not all entries of $\vec{A}$ can be selected at every time instant. Instead, the only valid actions are those which render $0 \leq \chi_\text{d} \leq 2^s - 1$.} from $\vec{A}$ and $\chi_\text{d}$ is the decimal representation of $\chi_\text{b}$.
In other words, an action consists of turning on/off a number of base stations, depending on the current status and on $A_z$.
This ensures that, instead of checking all possible statuses of SCs, we are taking only some samples of the entire set.

To illustrate the proposed idea, consider the following example. 
In a network with 4 SCs at time $t$, let $\vec{\chi}(t) = \{1, 1, 0, 1\}$. 
In this case, $\chi_\text{d}(t) = 13$.
Now, let us assume that $\xi = 2$ and that the action selected is $+\xi^1$.
In this case, the next status will be $\chi_\text{d}(t+1)= \chi_\text{d}(t) + \xi^1 = 15$, or in other words $\vec{\chi}(t+1) = \{1, 1, 1, 1\}$, which implies that all BSs are turned on at time $t+1$.

\subsubsection{Features}
The features used by the MC in order to find $\hat{Q}$ are the total power consumption of the network and the total load factor of each BS, such that
\begin{equation}
    \vec{X} = [P, \lambda_1, \lambda_2, \ldots,\lambda_{s+1}].
\end{equation}

\subsubsection{Cost}
The cost, or penalty, that we propose is described as the total power consumption plus a penalty---which is proportional to the number of SCs and the MC load---if the load of the MC is exceeded.
We can formally define it as
\begin{equation}\label{eq:vfa_cost}
    C = P + s\sigma\kappa \lambda_1,
\end{equation}
where $C$ is the cost, $\kappa$ is a penalty factor, and $\sigma\in\{0,1\}$ indicates whether or not the MC is overloaded.
\subsubsection{Complexity}
The most important benefit of the proposed solution is to greatly reduce the complexity of finding a good operating point.
Since there are $s$ SCs which could be switched off/on in any combination, an exhaustive search approach would have complexity $O(2^s)$, while our proposed approach only keeps track of $2s$ actions, and therefore has a complexity $O(s)$.

However, note that our solution does not strictly guarantee the constraint in~\eqref{eq:opt}.
This is important to give the RL algorithm the chance to explore different actions and learn what is not good.
By incorporating $\sigma\kappa$ into the cost, we can influence the RL algorithm not to violate the constraint, and thus satisfy the QoS requirements of the users whilst seeking the best policy regarding energy consumption.

\section{Performance Evaluation}\label{sec:performance}
In this section, the employed data set, the benchmark methods developed, performance evaluation metrics, and the obtained results will be discussed comprehensively.
The simulation parameters are provided in Table~\ref{tab:sim_parameters}, while Table~\ref{table:power_profile} presents the power consumption characteristics of the BSs used in the simulations.
\begin{table}[htbp]
\centering
\caption{Simulation Parameters}\label{tab:sim_parameters}
\begin{tabular}{ll}
\toprule
\textbf{Parameter}                   & \textbf{Value} \\
\midrule
 
\multicolumn{2}{c}{\textbf{SARSA with VFA}}      \\
Chance of choosing random action, $\epsilon$                            & 0.8            \\
Learning rate, $\alpha$                              & $10^{-7}$      \\
Discount factor, $\gamma$                              & 0.9            \\
Inter-space between two consecutive actions, $\xi$                                 & 2              \\
Minimum number of iterations for stopping, $j_\text{min}$                        & 10             \\
Threshold for the feature scaled cost, $\Omega$                              & $5\times10^{-2}$           \\
Consecutive iterations ensuring stopping criteria, $j_\text{rep}$                        & 10             \\
Maximum number of iterations                  & 100  \\  
\multicolumn{2}{c}{\textbf{General}}      \\
Number of time slots                  & 144            \\
Number of days                        & 1              \\
Number of grids considered for the MC & 2              \\
Number of grids considered per SC     & 1              \\
Bandwidth for the MC                  & 20 MHz         \\
Bandwidth per SC                      & 20 MHz         \\

\bottomrule
\end{tabular}
\end{table}
\subsection{Data Set}\label{sec:dataset}
In order to calculate the power consumption through~\eqref{eq:earthPower}, $\lambda$ is required for each BS.
In this regard, to obtain $\lambda$ values, we use a real call detail records~(CDR) data set provided by Telecom Italia\footnote{Data is available online at https://dandelion.eu/datamine/open-big-data/.}, in which the city of Milan is divided into 10,000 square-shaped grids with a dimension of 235x235 meters.
Then, within each grid, user call, text message, and internet activity levels were recorded with 10-minute resolution for a two months period~(November and December 2013).
Even though the data set consists only of unitless activity level values and there is no information provided regarding the data processing phase, the activity levels can be interpreted as grid-wise relative traffic loads, since they represent the volume of user-mobile network operator interaction at each time slot.
In the data processing phase of this work, first, we combine aforementioned separated activity levels~(i.e., call, text message, and internet).
Then, we randomly\footnote{The reason for choosing the grids in a random manner is to avoid similar activity level characteristics of the grids that have close proximity.} pick two grids for the MC and one grid for each SC.
Note that the activity levels at the two grids selected for the MC are further combined to create a traffic load data.
After that, all the activity levels are normalized together between 0 and 100, and the obtained values are treated as traffic loads for each cell.
\begin{table}[h]
    \centering
    \caption{Power profiles for BSs~\cite{Auer2011}}\label{table:power_profile}
    \begin{tabular}{ccccc}
        \toprule
        \multirow{3}{*}{BS Type} & \multirow{2}{*}{Efficiency}& \multicolumn{3}{c}{Power Consumption [W]} \\
                                 & & Transmit  & Operational  & Sleep   \\
                                 & $\eta_j$ & $P_{\text{T},j}$ & $P_{\text{o},j}$ & $P_{\text{s},j}$ \\
                              \midrule
        Macro & 4.7 & 20   & 130 & 75  \\
        RRH   & 2.8 & 20   & 84  & 56  \\
        Micro & 2.6 & 6.3  & 56  & 39  \\
        Pico  & 4.0 & 0.13 & 6.8 & 4.3 \\
        Femto & 8.0 & 0.05 & 4.8 & 2.9 \\
        \bottomrule
    \end{tabular}
\end{table}

\subsection{Benchmarking}\label{sec:bencmark}
Four different benchmark methods are used for comparison purposes, and they will be elaborated individually in the following paragraphs.
\subsubsection{Sorting}\label{sec:sorting}
Inspired by~\cite{sort_1, sort_2}, the sorting algorithm is developed to compare the results of the proposed algorithm.
In this method, the SCs are sorted in ascending order based on their load factors, $\lambda$.
Then, they are switched off sequentially until there is no available capacity left at the MC, and the rest of the SCs are kept ON.
Given the power consumption profile in~\eqref{eq:earthPower} and the characteristics of different types of BSs in Table~\ref{table:power_profile}, the MC consumes more power than the SCs for the same value of $\lambda$.
As such, it is wiser to switch off a SC with smaller traffic load in order to save more energy.
This concept lies at the heart of the sorting algorithm, since it aims at minimizing the energy consumption of the network.
On the other hand, as the SC switching off is performed only when there is enough capacity at the MC, this method also guarantees the service of the users after the offloading process.
\subsubsection{All-ON Method}\label{sec:allon}
There is no switching implemented in this method, meaning that all the SCs are always kept ON.
Accordingly, no offloading is needed as well in this case provided that none of the SCs are switched off at any time.
Therefore, it can be inferred that there is no concern of QoS in this method, since all the users are served by the BSs~(either MC or SC) that they were associated in the first place.
\subsubsection{All-OFF Method}\label{sec:alloff}
In this methods, the SCs are always kept switched off and their data traffics are offloaded to the MC.
However, this method is performed blindly, meaning that the data traffic of the SCs are offloaded to the MC regardless of its available capacity.
This means that the users, which are normally served by SCs, are vulnerable to service disruptions, since there is no guarantee that they will be served by the MC. 
Even if the service is provided by the MC, the QoS would be reduced in case there are more users than the available capacity, and in that case the MC reduces the available resources for each user by certain amount in order to keep all the users served.
\subsubsection{Exhaustive Search}\label{sec:exhaustive}
Exhaustive search is a method that tries to find the best policy among the set of all possible switching options consisting of the OFF/ON states of the SCs.
In particular, given the available capacity of the MC as a constraint, this method searches for the option with the least energy consumption.
Hence, this method guarantees the service for each user in the case of offloading, which in turn prevents the QoS of the users from being violated.
Note that exhaustive search returns the optimum policy, and thus the objective of any algorithm should be to mimic it as much as possible.
\subsection{Performance Metrics}\label{sec:metrics}
In this section, the metrics, which are used to evaluate the performance of the proposed algorithm and the benchmark methods are presented.
\subsubsection{Gain}\label{sec:gain}
In this work, we are interested in the percentage gain on the total energy consumption compared to the all-ON method.
It is calculated as
\begin{equation}\label{eq:gain}
    G = \dfrac{E_\text{on}-E_x}{E_\text{on}}\times 100,
\end{equation}
where $E_\text{on}$ and $E_x$ are the total energy consumption in joules with all-ON method and with one of the other methods, such that
\begin{equation}
    E_x\in\{E_\text{on}, E_\text{es}, E_\text{sort}, E_\text{vfa}\},
\end{equation}
where $E_\text{off}$, $E_\text{es}$, $E_\text{sort}$, and $E_\text{vfa}$ are the total energy consumption in joules with all-OFF, exhaustive search, sorting, and the proposed VFA-based methods, respectively.
\subsubsection{Power Consumption}\label{sec:powerconsumption}
Power consumption in watts during the simulation time are obtained for each method.
This is a beneficial metric to evaluate the performance of the methods, since it reflects the variations in power consumption for different times of a day.
Moreover, given that the gain is calculated on the energy consumption by accumulating the power consumption during the simulations, which can also be interpreted as upsampling, the detailed behaviours of the developed methods are kind of lost.
Thus, power consumption is also a utilitarian metric that paves the way for detailed behavioural observations.
\subsubsection{Average Network Throughput}\label{sec:tput}
The activity levels provided in the data set are assumed to be throughput demands with a unit of Mbps.
As such, $\lambda_i$ is treated as the normalized throughput demand from cell $i$, and then the holistic network throughput is calculated by combining the normalized throughput demands, $\bar{T}$, from each cell as follows:
\begin{equation}\label{eq:tput}
    \bar{T}(t)=\sum_{i=1}^{s+1}\lambda_i.
\end{equation}

However, there is one caveat that since the backhaul capacity of the cells is limited by the installed backbone, and thus cannot be exceeded, a BS $B_i$ penalizes the throughput for each user by $\Upsilon_i$ when the combined demand of the users exceeds $T_{\text{p},i}$, the maximum installed capacity of $B_i$, such that
\begin{equation}\label{eq:user_resource}
    \hat{r}_{\text{u},i}(t) = r_{\text{u},i}(t) - \Upsilon_i(t),
\end{equation}
where $r_{\text{u},i}(t)$ and $\hat{r}_{\text{u},i}(t)$ are the average user throughputs for users allocated by $B_i$ at time $t$ before and after penalization, respectively.

The penalty, $\Upsilon_i$, is calculated as 
\begin{equation}\label{eq:tput_penalty}
    \Upsilon_i(t)  = 
    \left\{
        \begin{array}{cc}
            \dfrac{T_{\text{r},i}(t)-T_{\text{p},i}}{N_{\text{u},i}(t)}, & T_{\text{r},i}(t) >T_{\text{p},i}\\
            0,                            &   \text{otherwise,}
         \end{array}
    \right.
\end{equation}
where $N_{\text{u},i}(t)$ is the number of users served by $B_i$ at time $t$ and
\begin{equation}\label{eq:resource_fact}
    T_{\text{r}, i}(t) = r_{\text{u},i}(t) N_{\text{u},i}(t) 
\end{equation}
is the combined required throughput of users allocated by $B_i$ at time $t$.
Next, as explained in Appendix~\ref{app:tput}, \eqref{eq:tput} can be manipulated and normalized, resulting in
\begin{equation}\label{eq:tput_for_ref}
    \bar{T}(t)=\sum_{i=1}^{s+1}u(-\lambda_i(t)+1)\lambda_i(t) + u(\lambda_i(t)-1),
\end{equation}
where $\bar{T}(t)$ is the normalized throughput of the network and $u(\cdot)$ is the unit step function, such that
\begin{equation}
 u(x)  = 
    \left\{
        \begin{array}{cc}
            1, & x\geq0\\
            0, &  x<0.
         \end{array}
    \right.
\end{equation}
\subsection{Scenarios}
The developed benchmark methods and the proposed VFA-based switching algorithm are tested in two different scenarios, namely Scenario A and Scenario B.
\subsubsection{Scenario A}
This is a simplistic scenario, where there is only one type of SC~(micro) in the network.
Moreover, the sleep mode power consumption for the SCs are assumed to be zero, such that the SCs are not contributing to the total power consumption of the network at all when they are switched off.
\subsubsection{Scenario B}
This is a complex scenario, comprising four different types of SCs, e.g. micro, remote radio head~(RRH), pico, and femto, are deployed in the network, and the number of SCs are distributed in these four types almost\footnote{We are changing the number of SCs in the network to observe the impacts of SC volume on the performance, and it is not always possible to distribute them equally, since the amount of SCs are sometimes not divisible by four, which is the number of SC types in the network. For example, when the number of SCs are 13, then the distribution becomes 3,3,3,4.} equally.
Moreover, the sleep mode power consumption is not assumed to be zero in this scenario, instead the values in Table~\ref{table:power_profile} are used.
Therefore, this scenario is more realistic than Scenario A, as in real networks, there are heterogeneous combinations of SCs and the sleep mode power consumption is not zero.
\subsection{Results}
\begin{figure}
	\centering
	\subfloat[Results for Scenario A when $\kappa=20$.]{%
		\resizebox{.5\columnwidth}{!}{
%
%
\definecolor{mycolor1}{rgb}{0.00000,0.44700,0.74100}%
\definecolor{mycolor2}{rgb}{0.85000,0.32500,0.09800}%
\definecolor{mycolor3}{rgb}{0.92900,0.69400,0.12500}%
\definecolor{mycolor4}{rgb}{0.49400,0.18400,0.55600}%
\definecolor{sort}{rgb}{0.46600,0.67400,0.18800}%
\begin{tikzpicture}

\begin{axis}[%
width=4.521in,
height=3.566in,
at={(0.758in,0.481in)},
scale only axis,
unbounded coords=jump,
xmin=4,
xmax=28,
xtick = {1,4,8,12,...,28},
xticklabels={1,4,8,12,...,28},
xlabel style={font=\color{white!15!black}},
xlabel={Number of small cells, $s$},
ymin=30,
ymax=77.8321507914959,
ylabel style={font=\color{white!15!black}},
ylabel={Gain on energy consumption, $G$ [\%]},
axis background/.style={fill=white},
xmajorgrids,
ymajorgrids,
legend style={at={(0.01,0.99)}, anchor=north west, legend cell align=left, align=left, draw=white!15!black},
label style={font=\Large},
ticklabel style={font=\large},
legend style={font=\Large}
]
\node[draw=black,minimum width = 5.93cm, minimum height = 30cm,fill=gray!20,opacity=.5] at (axis cs: 8.8,10){};
\addplot [color=mycolor1, line width=2.0pt, mark=*, mark size=4.5pt, mark options={solid, fill=white}]
  table[row sep=crcr]{%
1	17.5006539158714\\
4	40.3714067422681\\
8	53.5358128766169\\
12	60.9740531650912\\
16	65.2679510395927\\
20	68.7649796280913\\
24	72.833291858101\\
28	75.50577413708\\
30	76.6590064567365\\
};
\addlegendentry{All-OFF}

\addplot [color=mycolor3, line width=9.0pt]
  table[row sep=crcr]{%
1	16.5696861105565\\
4	34.0115700390562\\
8	42.4003996899367\\
12	45.8136199001857\\
15	45.7808787981661\\
16	nan\\
17	nan\\
18	nan\\
19	nan\\
20	nan\\
21	nan\\
22	nan\\
23	nan\\
24	nan\\
25	nan\\
26	nan\\
27	nan\\
28	nan\\
29	nan\\
30	nan\\
};
\addlegendentry{ES}

\addplot[color=mycolor4, line width=2.0pt, mark=square*, mark size=5pt, mark options={solid, fill=white}]
  table[row sep=crcr]{%
1	17.1782365511931\\
4	33.6599729685753\\
8	42.0216001601075\\
12	44.9574090488135\\
16	44.3304004490467\\
20	46.4035776670887\\
24	49.8962561276734\\
28	51.4426828688249\\
30	52.0519774114589\\
};
\addlegendentry{VFA}

\addplot [color=sort, line width=2.0pt, mark=diamond*, mark size=6pt, mark options={solid, fill=white, fill opacity=.8}]
  table[row sep=crcr]{%
1	16.563695954852\\
4	34.0115700390562\\
8	42.4003316896783\\
12	45.8136199001857\\
16	45.6218746177399\\
20	47.826734889299\\
24	51.354875058288\\
28	52.9497002088397\\
30	53.5759129347362\\
};
\addlegendentry{Sorting}

\node[anchor=north west, red, align=center] at(axis cs:15,60) (source) {\large ES \\ \large Terminated};
\draw[-latex, ultra thick, red, dashed, shorten >= 1.5em](source) -- (axis cs:15, 44.2707242975697);

\coordinate (a) at (axis cs:8,35);
\coordinate (b) at (axis cs:24,35);
\draw[->, >=latex, blue!20!white, line width=20pt, opacity=.9]   (a) to node[red]{\large Complexity increases} (b);

\end{axis}
\end{tikzpicture}%
}\label{fig:gain_sc_a}}
	\subfloat[Results for Scenario B when $\kappa=10$.]{%
		\resizebox{0.5\columnwidth}{!}{
%
%
\definecolor{mycolor1}{rgb}{0.00000,0.44700,0.74100}%
\definecolor{mycolor2}{rgb}{0.85000,0.32500,0.09800}%
\definecolor{mycolor3}{rgb}{0.92900,0.69400,0.12500}%
\definecolor{mycolor4}{rgb}{0.49400,0.18400,0.55600}%
\definecolor{sort}{rgb}{0.46600,0.67400,0.18800}%
\begin{tikzpicture}

\begin{axis}[%
width=4.521in,
height=3.566in,
at={(0.758in,0.481in)},
scale only axis,
unbounded coords=jump,
xmin=4,
xmax=28,
xtick = {4,8,12,...,28},
xticklabels={4,8,12,...,28},
ytick = {-5,0,5,10,...,20},
yticklabels={-5,0,5,10,...,20},
xlabel style={font=\color{white!15!black}},
xlabel={Number of small cells, $s$},
ymin=-5,
ymax=20,
ylabel style={font=\color{white!15!black}},
ylabel={Gain on energy consumption, $G$ [\%]},
axis background/.style={fill=white},
xmajorgrids,
ymajorgrids,
legend style={at={(0.01,0.99)}, anchor=north west, legend cell align=left, align=left, draw=white!15!black},
label style={font=\Large},
ticklabel style={font=\large},
legend style={font=\Large}
]
\node[draw=black,minimum width = 5.93cm, minimum height = 30cm,fill=gray!20,opacity=.5] at (axis cs: 8.8,10){};
\addplot [color=mycolor1, line width=2.0pt, mark=*, mark size=4.5pt, mark options={solid, fill=white}]
  table[row sep=crcr]{%
1	0.374332183467899\\
4	-3.05277958997035\\
8	-2.11371034764244\\
12	0.5123287744183\\
16	2.14151936566089\\
20	4.06762914114043\\
24	7.28758914420393\\
28	9.26351346620407\\
30	10.6196061552409\\
};
\addlegendentry{All-OFF}

\addplot [color=mycolor3, line width=9.0pt]
  table[row sep=crcr]{%
1	2.83543656250737\\
4	7.55489594666953\\
8	11.319246739909\\
12	12.9809193957307\\
15	14.0660945722313\\
16	nan\\
17	nan\\
18	nan\\
19	nan\\
20	nan\\
21	nan\\
22	nan\\
23	nan\\
24	nan\\
25	nan\\
26	nan\\
27	nan\\
28	nan\\
29	nan\\
30	nan\\
};
\addlegendentry{ES}

\addplot [color=mycolor4, line width=2.0pt, mark=square*, mark size=5pt, mark options={solid, fill=white}]
  table[row sep=crcr]{%
1	2.80746071304725\\
4	7.4569933048607\\
8	11.1352057830913\\
12	12.7799492076566\\
16	13.660195676892\\
20	14.6139715958432\\
24	15.9348142717074\\
28	16.8449130024169\\
30	17.342177206423\\
};
\addlegendentry{VFA}

\addplot  [color=sort, line width=2.0pt, mark=diamond*, mark size=6pt, mark options={solid, fill=white}]
  table[row sep=crcr]{%
1	1.18791070951607\\
4	2.624356203228\\
8	6.08223371001838\\
12	7.88566677625403\\
16	8.98396592112328\\
20	10.2183397398154\\
24	11.9668560147567\\
28	12.9450217422676\\
30	13.4173230746195\\
};
\addlegendentry{Sorting}

\node[anchor=north west, red, align=center] at(axis cs:15,20) (source) {\large ES\\ \large Terminated};
\draw[-latex, ultra thick, red, dashed, shorten >= 1em](source) -- (axis cs:15, 14.0660945722313);

\coordinate (a) at (axis cs:8,-2.5);
\coordinate (b) at (axis cs:24,-2.5);
\draw[->, >=latex, blue!20!white, line width=20pt, opacity=.9]   (a) to node[red]{\large Complexity increases} (b);

\end{axis}
\end{tikzpicture}%
}\label{fig:gain_sc_b}}
	\caption{Percentage gain performances compared to the all-on method for Scenario A and Scenario B.}\label{fig:gain}
\end{figure}
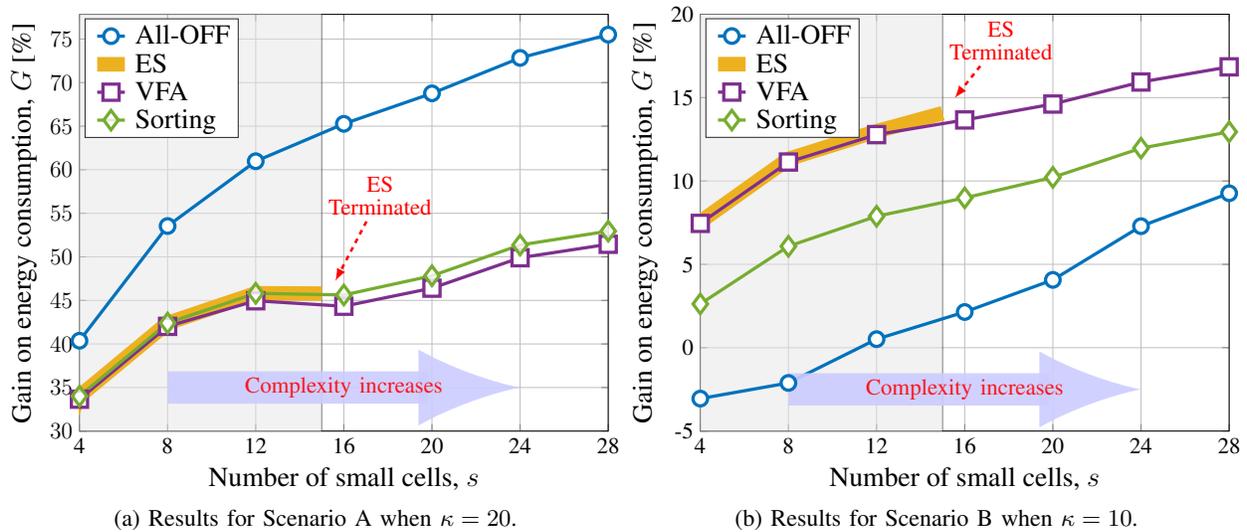
As explained in Section \ref{sec:gain}, Fig.~\ref{fig:gain} demonstrates how much gain in energy consumption is obtained when all-OFF, exhaustive search, sorting, and proposed VFA-based methods are compared to all-ON method.
Fig.~\ref{fig:gain_sc_a} shows the gain results for Scenario A, while the results for Scenario B are presented in Fig.~\ref{fig:gain_sc_b}.
Note that since the exhaustive search method is computationally demanding with $O(2^s)$, where it doubles the elapsed time when $s$ is incremented by 1, it is allowed to run only until $s=15$ for both Fig.~\ref{fig:gain_sc_a} and Fig.~\ref{fig:gain_sc_b}.
The idea in Fig.~\ref{fig:gain} is that we first compare the performances of all-OFF, sorting, and the proposed method to the exhaustive search, which is optimum, in a smaller network due to the time requirement for exhaustive search when $s$ increases.
Then, once we have an idea about the performance of the methods compared to the optimum, we terminate the exhaustive search and keep the other methods running in order to observe further behaviours in larger scale networks.

It can be seen in Fig.~\ref{fig:gain_sc_a} that the sorting method is working exactly same with the exhaustive search, while the proposed method follows them quite closely.
These outcomes can lead to a conclusion that the proposed method is outperformed by a more simpler and straightforward algorithm, and one can question the validity of the proposed method.
However, when Fig.~\ref{fig:gain_sc_a} is reconsidered together with Fig.~\ref{fig:gain_sc_b}, it can be inferred that the better performance of the sorting algorithm is not generalized.
While it gives promising results in simplistic Scenario A, it starts under-performing in the complex and realistic Scenario B.
On the other hand, given that the proposed VFA-based method gives very close results to the exhaustive search in both scenario, it seems quite immune to the changes in the scenario with generalized good performance.
In other words, although there might be simpler alternatives to the proposed method in simple scenarios, which are mostly unrealistic, the proposed method takes the advantage of being capable of generalization and works properly even when the scenario becomes more complicated and realistic owing to the nature of the developed VFA-based RL algorithm.
This, in turn, makes the proposed method work properly regardless of the conditions, while the sorting method, for example, relies on the simplicity of the scenario, making it impractical for realistic scenarios, where the assumptions in Scenario A are no more valid.

Another point about Fig.~\ref{fig:gain} is that the gain decreases significantly when the scenario is switched from A to B.
For the proposed method, for instance, the gain drops from around 52\% to 17\% when $s=30$, which yields around 67\% reduction.
This is mainly due to more heterogeneity of Scenario B, where there are four types of SCs.
Scenario A includes only micro cell, which is the second most energy consuming SC after RRH according to Table~\ref{table:power_profile}, making the SCs in Scenario A consume a considerable amount of energy.
For Scenario B, on the other hand, there are four types of SCs with distinctive power profiles, and thus total power consumption decreases because of the inclusion of pico and femto cells, which consume small amount of energy, in the network.
Besides, while SCs still consume energy when they are switched off in Scenario B, the sleep mode power consumption is assumed to be zero for Scenario A, making it consume overall less energy when switching is performed.

The last point that is worthy to discuss about the findings in Fig.~\ref{fig:gain} is related to the all-OFF method.
While it outperforms all the other method in terms of gain in Fig.~\ref{fig:gain_sc_a}, it becomes the worst-performing method in Fig.~\ref{fig:gain_sc_b}.
The reason behind this phenomenon is again the characteristic diversity between Scenarios A and B.
There is only micro cell, which is demanding in energy, in Scenario A, and thus switching off SCs almost always result in less energy consumption, whereas, due to the heterogeneity of Scenario B, the optimal policies for switching off are different from each other for each type of SC.
In particular, \eqref{eq:theorem} holds for larger $\lambda_j,~ j>1$ values when the type of SC goes from femto cell to RRH.
This means that the number of cases, where switching off is profitable, is larger for micro cell than that of femto and/or pico cells.
Therefore, since the all-OFF method switches off all the available SCs regardless of their types, the overall process becomes less profitable in Scenario B when compared to Scenario A.

\begin{figure}
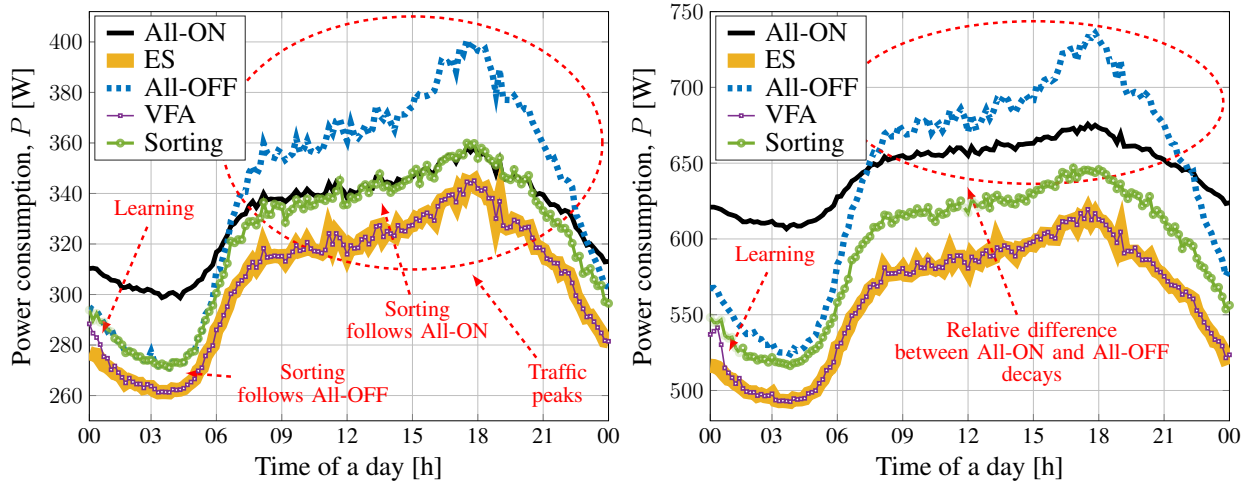
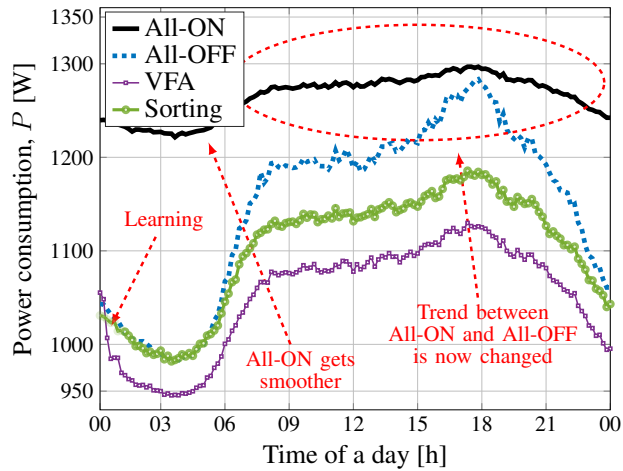

	\centering
	\subfloat[Averaged power consumption results over 25 rounds when $s=4$.]{%
		\resizebox{0.5\columnwidth}{!}{\input{images/power_sc_4_avg.tikz}}\label{fig:power_sc_4}}
	\subfloat[Averaged power consumption results over 25 rounds when $s=12$.]{%
		\resizebox{0.5\columnwidth}{!}{\input{images/power_sc_12_avg.tikz}}\label{fig:power_sc_12}}
		
	\subfloat[Averaged power consumption results over 25 rounds when $s=28$.]{%
		\resizebox{0.5\columnwidth}{!}{\input{images/power_sc_28_avg.tikz}}\label{fig:power_sc_28}}
	\caption{Power consumption performances of the developed methods for various number of SCs.}\label{fig:power_cons}
\end{figure}

Fig.~\ref{fig:power_cons} reveals the power consumption behaviours of the developed methods for various numbers of SCs included in the network.
Similar to the results in Fig.~\ref{fig:gain}, for consistency, we present the power consumption results for integer multiples of four, where the types of SCs are distributed equally.
Hence, $s$ value is altered as $s=4$, $s=12$, and $s=28$, respectively.
Note that Fig.~\ref{fig:power_sc_28} does not include the exhaustive search method, since it is allowed to run until $s=15$ for computational complexity reasons.

Fig.~\ref{fig:power_sc_4} demonstrates the power consumption performance of the developed methods when $s=4$.
The findings suggest that the proposed VFA based algorithm manages to mimic the exhaustive search almost perfectly apart from the initial learning phase.
Given the computational complexities for the exhaustive search and the proposed method are $O(2^s)$ and $O(s)$, respectively, these results confirm that the proposed method performs quite well~(i.e., producing near-optimal results) with drastic decrease in computational complexity.

It is also worth discussing the learning phase of VFA algorithm, which is common in all the three cases of $s$.
This behaviour is expected as we propose an online learning framework, where VFA-based cell switching method is deployed without any prior knowledge, and thus it learns by interacting with the real environment.
In other words, the promising performance after the initial learning phase is due to the experience obtained during the training.
The slightly worse initial performance is due the fact that VFA takes more random actions in the beginning in order to increase the knowledge about the environment.
This process is referred to as \textit{exploration}.
Then, after the exploration, the randomness in the actions taken decreases with the number of episodes in order to let the VFA start using the information it received, which is known as \textit{exploitation}.
One important point here is that the exploration process in the developed model takes a short amount of time, making the online implementation feasible, since longer learning phases would undermine the advantage of the VFA based solutions.
The reason behind preferring the online implementation over the offline one is that the former case is model-free, where it does not require any prior knowledge, while a full environmental knowledge is needed in the latter.
This, in turn, renders the online implementation more practical.
Therefore, even though the offline implementation is free from the possible negative impacts of the training process, it is comparatively less functional in real scenario, where full prior knowledge is often inaccessible.

Another interesting aspect that can be deduced from Fig.~\ref{fig:power_sc_4} is the behaviour of the sorting method.
It mirrors the all-OFF method in the beginning, where the data traffic is relatively low, whereas it starts following the all-ON method when the traffic load increases.
This is because the sorting method manages to switch off most~(even all) SCs, which is similar~(same) behaviour with the all-OFF method, when the traffic load is lighter, since there is more offloading opportunities in the MC owing to the low traffic patterns.
Moreover, the amount of offloaded data is also comparatively less at these times.
When the traffic volume becomes higher, on the other hand, since the amount of data to be offloaded and the occupancy of the MC increases simultaneously, the switching off becomes much harder~(even impossible), which is quite similar to the behaviour of the all-ON method.

Fig.~\ref{fig:power_sc_12} presents the power consumption results when $s=12$.
Similar to Fig.~\ref{fig:power_sc_4}, we observe the sub-optimal results of VFA during the training phase.
However, other than being short in time, the appreciable thing about this training phase is that the results that VFA produces are still reasonable even though it is not optimal.
Another interesting observation is the relative difference between all-ON and all-OFF methods shrinks compared to the one in Fig.~\ref{fig:power_sc_4}.
Taking into account the peak points, the relative difference between the all-OFF method and all-ON method decreases by around 25\% when $s$ increases from 4 to 12.

Two important questions arise from these findings.
First, why does the all-OFF method result in more power consumption than the all-ON method?
It is counter-intuitive to observe such results where switching all the SCs off causes more power consumption than always keeping all the SCs ON.
Moreover, we observe that all-ON outperforms all-OFF especially when the traffic loads are higher.
The reason behind this is that, as repeated previously, considering~\eqref{eq:theorem} together with Table~\ref{table:power_profile}, it is usually non-profitable to switch off SCs when the traffic load is above some certain threshold, which is different for each type of SC.
Therefore, it is not a rule of thumb that the switching off is always resulting in less power consumption.
The outcomes in Fig.~\ref{fig:gain_sc_b}, where the all-OFF method gives negative gains, also confirm this conclusion.
Nonetheless, this is not the only condition that makes the all-ON more favorable than the all-OFF method in terms of power consumption.
It is also the intensity of the SCs in the network.
Since $s$ is not large enough in both Figs.~\ref{fig:power_sc_4} and \ref{fig:power_sc_12}, the contribution of the SCs to the power consumption is comparatively less than that of the MC, therefore, the overall energy saving resulting from switching off SCs cannot prevail against the loss caused by offloading traffic to the MC.

Second, why does the relative difference between the all-OFF and all-ON methods decay when $s$ rises up from 4 to 12?
The answer for this question is related to the last discussion for the previous question; since the intensity of the SCs increases in the network with increasing $s$, the dominance of the MC in the total power consumption scales down.
This subsequently renders the gain resulting from switching off more significant, and thus the all-OFF method starts being more reasonable.
Hence, the number of instances that all-ON outperforms the all-OFF method also decreases when $s$ is increased from 4 to 12.
The results in Fig.~\ref{fig:gain_sc_b} again supports this conclusion, as the percentage gain enhances with increasing $s$.

Fig.~\ref{fig:power_sc_28} showcases the power consumption results when $s=28$.
It is again worth noting that exhaustive search is not included this time, since it is run until $s=15$, owing to the computational complexity concerns.
Unlike the results in Figs.~\ref{fig:power_sc_4} and Fig.~\ref{fig:power_sc_12}, there is no point in Fig.~\ref{fig:power_sc_28}, where all-OFF outperforms all-ON.
Similar to the previous discussions on this topic, the distribution of the total power consumption among the SCs and the MC is an integral aspect of the performances of the all-ON and all-OFF methods.
Given that the SCs now consume considerable amount of energy due to their increased number, the gain obtained from switching off the SCs, which are profitable to switch off according to \eqref{eq:theorem}, prevails over the loss incurred by switching off the SCs, which are non-profitable to switch off.

Another interesting point about the findings in Fig.~\ref{fig:power_sc_28} is that the power consumption for the all-ON method becomes smoother compared to the results in Figs.~\ref{fig:power_sc_4} and \ref{fig:power_sc_12}.
The relative peak-to-peak difference, for example, was around 20\% when $s=4$, whereas it drops to 6.2\% when $s$ increased to 28.
This is again because of the MC loosing its dominance in the total power consumption.
While, when fully-loaded, 72.3\% of the power consumption comes from the load dependent part for the MC, this rate is around 28\% on average~(minimum: femto cell with 8.3\%, maximum: RRH with 66.7\%) for the SCs.
In other words, the MC consumes more on the load dependent part, whereas SCs consume more on the static power.
Therefore, for the smaller $s$ values, the load dependent power consumption is higher as the MC is the main contributor to the total power consumption, while the load dependent power consumption becomes relatively less and the static power gets more significant for higher values of $s$.

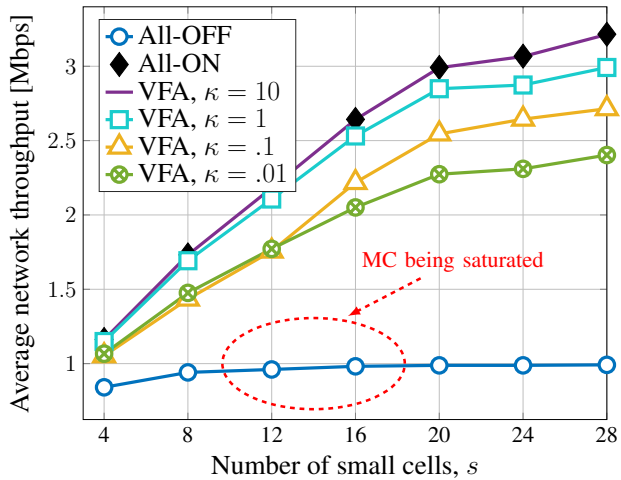
\begin{figure}
    \centering
    \resizebox{0.5\columnwidth}{!}{
%
%
\definecolor{mycolor1}{rgb}{0.00000,0.44700,0.74100}%
\definecolor{mycolor2}{rgb}{0.85000,0.32500,0.09800}%
\definecolor{mycolor3}{rgb}{0.92900,0.69400,0.12500}%
\definecolor{mycolor4}{rgb}{0.49400,0.18400,0.55600}%
\definecolor{mycolorx}{rgb}{0.1,0.8,0.8}%
\definecolor{mycolor5}{rgb}{0.46600,0.67400,0.18800}%
\begin{tikzpicture}

\begin{axis}[%
width=4.521in,
height=3.566in,
at={(0.758in,0.481in)},
scale only axis,
xmin=3,
xmax=28,
xtick = {1,4,8,12,...,28},
xticklabels={1,4,8,12,...,28},
xlabel style={font=\color{white!15!black}},
xlabel={Number of small cells, $s$},
ymin=0.623847406713118,
ymax=3.4,
ylabel style={font=\color{white!15!black}},
ylabel={Average network throughput [Mbps]},
axis background/.style={fill=white},
xmajorgrids,
ymajorgrids,
legend style={at={(0.03,0.97)}, anchor=north west, legend cell align=left, align=left, draw=white!15!black},
label style={font=\Large},
ticklabel style={font=\large},
legend style={font=\Large}
]
\addplot [color=mycolor1, line width=2.0pt, mark=*, mark size=4.8pt, mark options={solid, fill=white}]
  table[row sep=crcr]{%
4	0.840695923043494\\
8	0.940944364318826\\
12	0.960662534110519\\
16	0.981829461086106\\
20	0.988309145553572\\
24	0.988316123049098\\
28	0.991674956180438\\
};
\addlegendentry{All-OFF}

\addplot [color=black, only marks, mark=diamond*, mark size=7.5pt, mark options={solid, fill=black}]
  table[row sep=crcr]{%
4	1.16233185099296\\
8	1.73134355473153\\
12	2.18279071138308\\
16	2.64269237300819\\
20	2.99187899797488\\
24	3.06801073025422\\
28	3.21539999778351\\
};
\addlegendentry{All-ON}

\addplot [color=mycolor4, line width=2.0pt]
  table[row sep=crcr]{%
4	1.16229733615803\\
8	1.73092404513224\\
12	2.1818545387318\\
16	2.64161222052045\\
20	2.99025010440592\\
24	3.06548270410538\\
28	3.21395323116818\\
};
\addlegendentry{VFA, $\kappa=10$}

\addplot [color=mycolorx, line width=2.0pt, mark=square*, mark size=5pt, mark options={solid, fill=white}]
  table[row sep=crcr]{%
4	1.14856594581061\\
8	1.69207123448963\\
12	2.10703166019912\\
16	2.5311341297155\\
20	2.84894110651153\\
24	2.87396171040514\\
28	2.99182031114209\\
};
\addlegendentry{VFA, $\kappa=1$}

\addplot [color=mycolor3, line width=2.0pt, mark=triangle*, mark size=7.5pt, mark options={solid, fill=white}]
  table[row sep=crcr]{%
4	1.05248141185505\\
8	1.4341128232636\\
12	1.75664459940772\\
16	2.21805576692296\\
20	2.54487712969302\\
24	2.6442842161662\\
28	2.71457823814678\\
};
\addlegendentry{VFA, $\kappa=.1$}

\addplot [color=mycolor5, line width=2.0pt, mark=otimes*, mark size=4.8pt, mark options={solid, fill=white}]
  table[row sep=crcr]{%
4	1.06684680116391\\
8	1.47571254970925\\
12	1.77195590599314\\
16	2.05054784456982\\
20	2.27388193587175\\
24	2.31007176839534\\
28	2.4028737167164\\
};
\addlegendentry{VFA, $\kappa=.01$}

\draw[ultra thick, red, dashed] (14, 1) ellipse (2cm and 1cm);
\node[anchor=north west, align=center, red] at(axis cs:16,1.8) (source) { \large MC being saturated};
\draw[-latex, ultra thick, red, dashed, shorten >= 2em](source) -- (axis cs:14, 1.2);

\end{axis}
\end{tikzpicture}%
}
    \caption{Results for normalized network throughput against the number of SCs. The network throughput is calculated by averaging out the obtained throughput values at each time slot during the simulation period.}\label{fig:tput}
\end{figure}
Fig.~\ref{fig:tput} shows normalized average network throughput, which is calculated through~\eqref{eq:tput_for_ref}, for various $\kappa$ values of VFA, the all-ON and the all-OFF methods.
Furthermore, the activity levels in the data set are assumed to be in Mbps after the pre-processing detailed in Section~\ref{sec:dataset}.
The objective of demonstrating these results is to highlight the impact of $\kappa$ value on the performance of the proposed VFA based switching algorithm.
In addition, the findings also display the cost of switching off all the SCs without taking into account the available capacity at the MC.
As such, the results suggest that there is an upper bound for the all-OFF method, since it only relies on the capacity of the MC, and after normalization, each BS~(MC or SC) has a capacity of unity at maximum.
Given that the MC is the only BS that is kept ON, all-OFF switches off all the SCs and offloads their traffic to the MC, meaning that it has one unit of capacity available in the network.
Therefore, these results also confirm that having a blind policy, that is, acting without considering the environmental conditions and/or constraints, is not a wise idea, since it results in the degradation of the QoS of users as well as being more costly in power consumption on some occasions, as already proven in Figs.~\ref{fig:gain_sc_b}, \ref{fig:power_sc_4}, and \ref{fig:power_sc_12}.
The purpose of presenting the results for the all-ON method is to demonstrate the best case~(upper bound) that can be achieved in terms of throughput.

As seen in Fig.~\ref{fig:tput}, we obtained promising results for the proposed algorithm.
It gives quite close results to the all-ON method, proving a good performance as it is producing similar results to the best case.
By considering the findings in Figs.~\ref{fig:gain}, \ref{fig:power_cons}, and \ref{fig:tput} together, it is easy to deduce that the proposed VFA based switching algorithm performs outstandingly well in terms of both power consumption and the throughput, since it reduces the power consumption~(similar to exhaustive search) without compromising on the QoS of the users~(similar to the all-ON method).
Provided that exhaustive search and all-ON are the best methods in terms of power consumption and QoS, respectively, the proposed method combines the advantages of both methods.

Fig.~\ref{fig:tput} also showcases the impact of $\kappa$ in \eqref{eq:vfa_cost} on the performance of the developed VFA model.
The results suggest that the throughput performance of the proposed VFA decreases with decreasing values of $\kappa$.
This is because higher values of $\kappa$ will result in more cost been incurred for the case when the demanded capacity exceeds the available capacity at the MC.
As explained in Section~\ref{sec:metrics}, when the demanded capacity is higher than the available one, the networks reduces the allocated bandwidth for each user in order to accommodate the demands of all users.
In this regard, the agent refrains from switching off a SC, whose demanded capacity is larger than the available one at the MC, which in turn helps to keep the QoS above the required level.
However, for the smaller values of $\kappa$, the agent starts following more relaxed policies on the given constraint, where it takes more actions that are against the above-mentioned demanded/available capacity criterion.
Hence, the obtained throughput starts decreasing for lower values of $\kappa$.
For the extreme scenario, where $\kappa=0$, then \eqref{eq:vfa_cost} becomes $C=P$, meaning that the agent only focuses on the total power consumption and does not care about the constraint of available capacity at the MC.
Intuitively, the agent would be reducing the total power consumption as much as possible at the expense of QoS degradation. 

\section{Conclusion}\label{sec:conclusion}
In this paper we presented a RL-based solution for cell switching, which is capable of learning the best policy in a dense HetNet environment in order to save energy and satisfy the QoS at the same time.
We evaluated our solution using real data from Milan, Italy and compared it with various benchmark methods.
The results in terms of power and energy consumption show that the proposed method can perform just as well as the exhaustive search method, which produces the optimum solutions, regardless of the complexity and size of the given scenario/environment.
Moreover, the proposed method resulted in much fewer computations than that of exhaustive search, meaning that it is a scalable method.
Furthermore, network throughput was also measured, and we observed that the proposed method gave similar results with the all-ON method, which is the best in terms of the QoS due to the fact that it does not include any offloading.
Therefore, the proposed VFA based cell switching method resulted in a significant reduction in the network's energy consumption without much compromise on the QoS, thus making it suitable for practical application.
In future, we plan to include traffic load predictions of the BSs as well as user mobility predictions in order to reduce the time spent in the switching process.

\appendices
\section{Normalized Network Throughput}\label{app:tput}
The throughput for $B_i$ at time $t$ can be expressed as the product of the average user throughput and the number of users it serves, as
\begin{equation}\label{eq:tput_with_number_of_users}
    T_i(t)=\hat{r}_{\text{u}}(t)N_{\text{u},i}.
\end{equation}

After that, using \eqref{eq:user_resource} into \eqref{eq:tput_with_number_of_users}, we get
\begin{equation}\label{eq:tput_user_derivation}
    T_i(t)=(r_{\text{u}}(t)-\Upsilon_i) N_{\text{u},i}.
\end{equation}

When \eqref{eq:tput_penalty} is used in \eqref{eq:tput_user_derivation}, $T(t)$ becomes
\begin{equation}\label{eq:tpu_cases_simplified}
\begin{split}
 T_i(t)  =  \,r_{\text{u},i}(t)N_{\text{u},i} - 
    \left\{
        \begin{array}{cc}
            T_{\text{r},i}(t)-T_{\text{p},i}, & T_{\text{r},i}(t)>T_{\text{p},i}\\
            0, &  \text{otherwise},
         \end{array}
    \right.
\end{split}
\end{equation}
after simplifying.

Then, using \eqref{eq:resource_fact}, \eqref{eq:tpu_cases_simplified} can be rewritten as 
\begin{equation}
     T_i(t)  = \\
    \left\{
        \begin{array}{cc}
            T_{\text{p},i}, & T_{\text{r},i}(t)>T_{\text{p},i}\\
            T_{\text{r},i}(t), &  \text{otherwise},
        \end{array}
    \right.
\end{equation}

Next, the throughput can be normalized with respect to the installed capacity, as
\begin{equation}\label{eq:norm}
    \lambda_i(t)=\frac{T_{\text{d},i}(t)}{T_{\text{p},i}},
\end{equation}
where $\lambda_i(t)$ is the normalized throughput---also defined as the load---of $B_i$ at time $t$. 
Therefore, dividing~\eqref{eq:tpu_cases_simplified} and using~\eqref{eq:norm}, we arrive at
\begin{equation}\label{eq:tput_last_before_step}
     \bar{T}_i(t)  = \\
    \left\{
        \begin{array}{cc}
            1, & \lambda_i(t)>1\\
            \lambda_i(t), &  0\leq \lambda_i(t)\leq 1,
         \end{array}
    \right.
\end{equation}
where $\bar{R}_{\text{d},i}(t)$ is the normalized version of $R_{\text{d},i}(t)$.
Then, using the unit step function, \eqref{eq:tput_last_before_step} can be rewritten as 
\begin{equation}
     \bar{T}_i(t)=u(-\lambda_i(t)+1)\lambda_i(t) + u(\lambda_i(t)-1).
\end{equation}

Lastly, in order to calculate the holistic network throughput, we perform a summation over all BSs, arriving at~\eqref{eq:tput_for_ref}.

\section*{Acknowledgement}
We acknowledge the support of EPSRC (GCRF) funds under the Grant no. EP/P028764/1.

\bibliographystyle{IEEEtran}

\end{document}